\documentclass[letterpaper, 10 pt, conference]{ieeeconf}  

\IEEEoverridecommandlockouts                              

\usepackage{amssymb,latexsym,amsfonts,amsmath}
\usepackage{graphicx}
\include{diagrams}
\usepackage{xcolor}
\usepackage{dsfont}

\newtheorem{theorem}{Theorem}[section]
\newtheorem{lemma}[theorem]{Lemma}
\newtheorem{problem}[theorem]{Problem}
\newtheorem{definition}[theorem]{Definition}
\newtheorem{remark}[theorem]{Remark}

\newcommand{\R}{{\mathbb{R}}}

\newcommand{\N}{{\mathbb{N}}}
\newcommand{\G}{{\mathcal{G}}}
\newcommand{\V}{{\mathcal{V}}}
\newcommand{\E}{{\mathcal{E}}}
\newcommand{\e}{\mathsf{e}}
\newcommand{\ie}{{\it i.e.}}

\DeclareMathOperator{\Tr}{Tr}
\DeclareMathOperator{\diff}{d}
\newcommand{\ra}{\rightarrow}

\newcommand{\ul}{\underline}
\newcommand{\ol}{\overline}

\newcommand{\EE}{\mathds{E}}

\title{\LARGE \bf
Distributed Consensus of Stochastic Multi-agent Systems with Prescribed Performance Constraints
}

\author{Pushpak Jagtap$^1$ and Dimos V. Dimarogonas$^2$ 
\thanks{This work was partially supported by the Wallenberg AI, Autonomous Systems and Software Program (WASP) funded by the Knut and Alice Wallenberg Foundation, the Swedish Foundation for Strategic Research (SSF), the Swedish Research Council (VR), the ERC CoG LEAFHOUND project, and Robert Bosch Center for Cyber-Physical Systems.} 
\thanks{$^1$Pushpak Jagtap is with Robert Bosch Center for Cyber-Physical Systems, Indian Institute of Science, Bangalore 560012, India {\tt\small pushpak@iisc.ac.in}}
\thanks{$^2$ Dimos V. Dimarogonas is with the Division of Decision and Control Systems, KTH Royal Institute of technology, SE-100 44 Stockholm, Sweden. {\tt\small dimos@kth.se}}%
}

\begin{document}

\maketitle
\thispagestyle{empty}
\pagestyle{empty}

\begin{abstract}

This paper focuses on the problem of distributed consensus control of multi-agent systems while considering two main practical concerns $(i)$ stochastic noise in the agent dynamics and $(ii)$ predefined performance constraints over evolutions of multi-agent systems. In particular, we consider that each agent is driven by a stochastic differential equation with state-dependent noise which makes the considered problem more challenging compare to non-stochastic agents. The work provides sufficient conditions under which the proposed time-varying distributed control laws ensure consensus in expectation and almost sure consensus of stochastic multi-agent systems while satisfying prescribed performance constraints over evolutions of the systems in the sense of the $q$th moment. Finally, we demonstrate the effectiveness of the proposed results with a numerical example. 
  
\end{abstract}

\section{Introduction}
The past decade has witnessed an ever-growing interest in the
study of multi-agent systems (MAS) due to their extensive applications in both science and engineering, see \cite{dorri2018multi,darmanin2017review}, and references therein for examples. Among many interesting problems, the synchronization or consensus problem of multi-agent system is an active research topic in the past years \cite{mesbahi2010graph, olfati2004consensus}, whose objective is to design a distributed consensus algorithm (or protocol) using only limited neighborhood information to ensure that all the agents achieve some common control objective, such as convergence to a common state.

In practice, stochastic disturbances, such as thermal noise,
channel fading, quantization effect during encoding and decoding, are inevitable and cannot be avoided in real-world systems. Therefore, the consensus problems of stochastic multi-agent systems (SMAS) have attracted much attention. Since the traditional consensus definition is not applicable in a stochastic setting, several researchers have proposed several consensus conceptions in different
probabilistic senses such as mean-square consensus, consensus in $q$th moment, consensus in probability, and
almost sure consensus.  Examples of few results include the mean-consensus protocol for SMAS \cite{wang2014decentralised}, consensus in probability for discrete-time SMAS \cite{ding2015event}, mean-square consensus control for time-varying SMAS \cite{ma2016event}, exponential consensus of SMAS with delay \cite{tang2015leader}, exponential leader-follower consensus \cite{sakthivel2018leader}, approximate consensus of discrete-time SMAS \cite{amelina2014approximate}, and average consensus of SMAS in $q$th moment \cite{liu2013distributed}. 

It is worth mentioning that the quality of performance, such as maximum overshoot, rate of convergence, and steady-state error, is usually required to be satisfied in practical systems. By considering such performance constraints, authors in \cite{bechlioulis2008robust} proposed \emph{prescribed performance control} (PPC) to ensure stability while respecting those constraints. Very recently many researchers adapted the PPC approach to guarantee prescribed performance constraints while solving various problems of multi-agent systems with deterministic agents \cite{macellari2016multi,chen2019consensus,stamouli2020multi,mehdifar2020prescribed}. On the other hand, there are a very few works available on utilizing PPC for stochastic control systems \cite{shao2018adaptive,sui2020novel}. However, as far as we know, there is no work available in the literature on the consensus of stochastic multi-agent systems while considering prescribed performance constraints.   

To the best of our knowledge, this paper is the first to address the consensus control of stochastic multi-agent systems with prescribed performance constraints. In this paper, we consider that the evolution of each agent is given by a stochastic differential equation with state-dependent noise and is driven by a conventional first-order consensus protocol \cite{olfati2004consensus} with an external control input. Further, for a given communication graph topology and predefined performance constraints, the paper proposes a time-varying distributed control law that guarantees consensus in considered SMAS along with the sufficient conditions over system and prescribed performance function parameters. In particular, we proposed results considering two well-know stochastic consensus notions: (i) consensus in expectation and (ii) almost sure consensus.  

The remainder of this paper is structured as follows. In Section II, we  introduce stochastic multi-agent systems and prescribed performance constraints. Then, we formally define the problem considered in
this paper. Section III provides the sufficient conditions under which the proposed distributed control law ensures stochastic consensus of SMAS while guaranteeing prescribed performance constraints. Section IV demonstrates the effectiveness of the results using a numerical example. Finally, Section V concludes the paper. 

\section{Preliminaries and Problem Statement}
\subsection{Notations}\label{II1}
Let the triplet $(\Omega, \mathcal{F}, \mathds{P})$ denote a probability space with a sample space $ \Omega $, a filtration $ \mathcal{F} $, and the probability measure $ \mathds{P} $. The filtration $\mathds{F}= (\mathcal{F}_s)_{s\geq 0}$ satisfies the usual conditions of right continuity and completeness \cite{oksendal2013stochastic}. Let $(W_s)_{s\geq 0}$ be a $\mathds{F}$-Brownian motion. We use $\EE[\cdot]$ to denote the expectation operator. The symbols $ \N $, $ \N_0 $, $ \R$, $\R^+,$ and $\R_0^+ $ denote the set of natural, nonnegative integer, real, positive, and nonnegative real numbers, respectively. We use $ \R^{n\times m} $ to denote a vector space of real matrices with $ n $ rows and $ m $ columns. We use $\|\cdot\|$ to represent Euclidean norm. For $a\in\R$, we denote absolute value of $a$ by $|a|$. For $a,b\in\R$ and $a< b$, we use $(a,b)$ to represent an open interval in $\R$. For $a,b\in\N$ and $a\leq b$, we use $[a;b]$ to denote a close interval in $\N$. We use $I_n$ and $0_{n}$ to denote identity matrix and zero matrix in $\R^{n\times n}$, respectively. A diagonal matrix in $\R^{n\times n}$ with diagonal entries $d_1,\ldots, d_n$ is denoted by $diag\{d_1,\ldots, d_n\}$. Given a matrix $M\in\R^{n\times m}$, $M^T$ represents transpose of matrix $M$. Given a matrix $P\in\R^{n\times n}$, $\Tr(P)$ represents the trace of matrix $P$, and $P>0$ and $P\geq0$ denote positive definite and semi-definite matrices, respectively. Given a set $A$, we use $|A|$ to represent the cardinality of the set $A$. We use notations $\mathcal{K}$ and $\mathcal{K}_\infty$
to denote different classes of comparison functions, as follows:
$\mathcal{K}=\{\psi:\mathbb{R}_{0}^+ \rightarrow \mathbb{R}_{0}^+ |$ $ \psi$ is continuous, strictly increasing, and $\psi(0)=0\}$; $\mathcal{K}_\infty=\{\psi \in \mathcal{K} |$ $ \lim\limits_{r \rightarrow \infty} \psi(r)=\infty\}$.
\subsection{Graph Theory}
An undirected graph \cite{mesbahi2010graph} is defined as $\G=(\V,\E)$ with the vertices set $\V=\{1,2,\ldots,n\}$ and the edges set $\E=\{(i,j)\in\V\times\V\mid j\in\mathcal{N}_i\}$, where $\mathcal{N}_i$ denotes the set of neighbouring agents of agent $i$ that can communicate with agent $i$. Let us index edges in set $\E$ as $e_1,e_2,\ldots,e_m$, where $m=|\E|$ is the number of edges in the graph. A \emph{path} is a sequence of edges that connects two different vertices. A graph is \emph{connected} if and only if there exist a path between any pair of vertices. A graph is a \emph{tree} if and only if there exist exactly one path between any pair of vertices. By assigning an arbitrary orientation to each edge of $\G$, we define the \emph{incidence matrix}, $D\in\R^{n\times m}$, with the rows of $D$ being indexed by the vertices and columns being indexed by edges; and the element $d_{ij}=1$ if the vertex $i$ is the head of the edge $(i,j)$, $d_{ij}=-1$  if the vertex $i$ is the tail of the edge $(i,j)$, and $d_{ij}=0$ otherwise. The \emph{graph Laplacian} of $\G$ is described as $L=DD^T$. In addition, $L_e=D^TD$ is the so-called \emph{edge Laplacian}.      
\subsection{System Description}
In this work, we consider a multi-agent system with $n$ first-order stochastic agents modeled by following stochastic differential equation:
\begin{align}\label{agent_dynamic}
\diff x_i=u_i \diff t+g(x_i)\diff {W}_t, \quad i\in[1;n],
\end{align}
where $x_i\in\R$, $u_i\in\R$ are the position and control input of the $i$th agent, respectively and $g:\R\ra\R$ is a Lipschitz continuous diffusion function with Lipschitz constant $\mathsf{k}_g\in\R_0^+$ such that: $\|g(x)-g(x')\|\leq \mathsf{k}_g\|x-x'\|$ for all $x,x'\in\R$. 
Note that for the sake of simplicity, here we consider one-dimensional agents and the result can be extended to higher dimensions with the appropriate use of Kronecker product.

Let us consider that agents represent the vertices $\V=\{1,2,\ldots,n\}$ of an undirected graph $\G$. We assume that the communication graph is static, i.e., neighbouring agents $\mathcal{N}_i$ of agent $i\in[1;n]$ do not vary over time and  
each agent is driven by a first-order consensus protocol with an external input $v_i\in\R$, \ie, $u_i=-\sum_{j\in\mathcal{N}_i}(x_i-x_j)+v_i$, with corresponding stochastic differential equation given as:
\begin{align}\label{follwer_SDE}
\diff x_i=\Big(-\sum_{j\in\mathcal{N}_i}(x_i-x_j)+v_i\Big) \diff t+g (x_i) \diff {W}_t, \quad i\in\V.
\end{align}

Let $x=[x_1,\ldots,x_n]^T\in \R^n$ and $v=[v_{1},\ldots,v_{n}]^T\in\R^{n}$ be stack vectors of absolute positions and external inputs of all agents, respectively. Denote $\ol{x}=[\ol x_1,\ldots\ol x_m]^T\in\R^m$ the stack vector of relative positions between the pair of communicating agents $(i,j)=e_k\in\E$, where $\ol x_k\triangleq x_{ij}= x_i-x_j$ such that $(i,j)=e_k\in\E, k\in[1;m]$. We also mention some interesting properties which are useful in the paper: $Lx=D\ol x$, $\ol x=D^T x$, and if $\ol x=0$, we have $Lx=0$. By stacking agents in \eqref{follwer_SDE}, the dynamics of the stochastic multi-agent system is rewritten as:
\begin{align} \label{MAS_dynamic}
\diff x=(-Lx+v)\diff t+G(x) \diff W_t,
\end{align}  
where $L$ is the graph Laplacian and $G(x):=[g(x_1),\ldots,g(x_n)]^T$.

Next, we introduce stochastic consensus notions for stochastic multi-agent systems which are adapted from \cite{huang2007stochastic}.
{\begin{definition}[Consensus in expectation]\label{con_exp}
	The agents in the stochastic multi-agent system \eqref{MAS_dynamic} are said to reach {consensus in expectation} if the following holds 
	\begin{align}
		\lim_{t\ra\infty}\mathbb{E}[\ol x(t)^T\ol x(t)]=0,
	\end{align} 
	where $\ol x$ is the stack vector of relative positions.
\end{definition}
\begin{definition}[Almost sure consensus]\label{con_as}
	The agents in the stochastic multi-agent system \eqref{MAS_dynamic} are said to reach {almost sure (a.s.) consensus} if the following holds  $\lim_{t\ra\infty}\ol x(t)=0,$ almost surely (i.e., $\mathbb{P}[\lim_{t\ra\infty}\ol x(t)=0]=1$).
\end{definition}}

For later use, we recall the infinitesimal generator (denoted by the operator $\mathcal{L}$) for a stochastic system $S: \diff x=f(x)\diff t+g(x) \diff W_t$, $x\in\R^n$ using It\^o's differentiation \cite{oksendal2013stochastic}. Let $V:\R^n\ra\R_0^+$ be a twice differentiable continuous function. The infinitesimal generator of $V$ associated with a stochastic system $S$ is an operator, denoted by $\mathcal{L} V$, and given by
\begin{align}\label{infinitesimal}
\mathcal{L} V(x)=\frac{\partial V}{\partial x}f(x)+\frac{1}{2}\Tr\Big(g(x)^T\frac{\partial^2 V}{\partial x^2}g(x)\Big),
\end{align}
for all $x\in \R^n$.
\begin{lemma}[\cite{oksendal2013stochastic}]\label{Lemma_Lyapunov}
	Consider a stochastic system $S: \diff x=f(x)\diff t+g(x) \diff W_t$, a twice differentiable continuous function $V:\R^n\ra\R_0^+$, and a constant $q\in\R^+$. If there exists a constant $\kappa\in\R^+$, a function $\ol\psi\in\mathcal{K}_\infty$, and a convex function $\ul\psi\in \mathcal{K}_\infty$ such that for all $x\in\R^n$ the following hold
	\begin{align}
	&\ul\psi(\|x\|^q)\leq V(x)\leq	\ol\psi(\|x\|^q),\label{abcd}\\
	&\mathcal{L}V(x)\leq-\kappa V(x),
	\end{align}
	then the solution of $S$ satisfies
	\begin{align}
	\EE[\|x(t)\|^q]\leq \ul\psi^{-1}(\ol\psi(\|x(0)\|^q)\e^{-\kappa t})
	\end{align}
	for all $t\in\R_0^+$.
\end{lemma}
{\begin{lemma}[\cite{wu2011stochastic}]\label{stoch_barbalat}
	For a stochastic system $S: \diff x=f(x)\diff t+g(x) \diff W_t$, if there exist a twice differentiable continuous function $V:\R^n\ra\R_0^+$ satisfying \eqref{abcd} and 
	$$\mathcal{L}V(x)\leq-\beta(x),$$
	where $\beta:\R^n\ra\R_0^+$ is continuous and nonnegative. Then, for each $x_0\in\R^n$, $\lim\limits_{t\ra\infty}\beta(x(t))=0$ a.s.
\end{lemma}}
\subsection{Prescribed Performance}
This subsection provides preliminary knowledge on the prescribed performance control (PPC) \cite{bechlioulis2008robust}. The aim of PPC is to prescribe the evolution of the relative position $\ol x_k(t)$ within some predefined region which can be expressed in the form of the following inequality 
\begin{align}\label{performane_bound}
-\rho_k(t)<\ol x_k(t)<\rho_k(t) 
\end{align}
for all $t\in\R_0^+$, where $\rho_k:\R_0^+\ra\R^+$, $k\in[1;m]$ are positive, smooth, and strictly decreasing performance functions that introduce the desired predefined bounds for the relative positions. In this work, we consider the following performance function 
\begin{align}\label{performane_function}
\rho_k(t)=(\rho_{k0}-\rho_{k\infty})\e^{-\epsilon_k t}+\rho_{k\infty},
\end{align} 
where $\rho_{k0}$, $\rho_{k\infty}$, and $\epsilon_k$ are positive constants with $\rho_{k0}>\rho_{k\infty}$ and $\rho_{k\infty}=\lim_{t\ra\infty}\rho_k(t)$ represents relative positions at steady state. Now by normalizing $\ol x_k$  with respect to the performance function $\rho_k$, we define the modulating error as $\hat x_k(t)=\frac{\ol x_k(t)}{\rho_k(t)}$ and the corresponding prescribed performance region $\hat{\mathcal{D}}_k:=\{\hat x_k\mid\hat x_k\in(-1,1)\}$. Then the modulated error is transformed through a transformation function $T_k:\hat{\mathcal{D}}_k\ra\R$ such that $T_k(0)=0$ and is chosen as
\begin{align}\label{transformation_fun}
T_k(\hat x_k)=\ln \Big(\frac{1+\hat x_k}{1-\hat x_k}\Big).
\end{align}
The transformed error is then defined as $\xi_k=T_k(\hat x_k)$. By differentiating $\xi_k$ with respect to time, we obtain transformed error dynamics as
\begin{align}
\dot{\xi}_k=\phi_k(\hat x_k,t)[\dot{\ol x}_k+\alpha_k(t)\ol x_k],
\end{align}
where $\phi_k(\hat x_k,t):=\frac{1}{\rho_k(t)}\frac{2}{1-\hat{x}_k^2}>0$ for all $\hat x_k\in(1,-1)$ and $\alpha_k(t):=-\frac{\dot{\rho}_k(t)}{\rho_k(t)}>0$ for all $t\in\R_0^+$ are the normalized Jacobian of the transformation function $T_k$ and the normalized derivative of the performance function $\rho_k$, respectively. Note that, since the evolution of $\ol x_k$ is stochastic, the transform error dynamics can also be written as a stochastic differential equation and the corresponding incremental form will be given in \eqref{errdyn}. It can be verified that if the transformed error is bounded, then the modulated error $\hat x_k$ is constrained within the region $\hat{\mathcal{D}}_k$. This further implies that the error $\ol x_k$ evolves within the predefined performance bound \eqref{performane_bound}. Since we are dealing with stochastic systems, we ensure the satisfaction of considered performance constraints in the sense of $q$th moment by showing the boundedness of $\mathbb{E}[\|\xi_k\|^q]$ (i.e. boundedness of $\xi_k$ in $q$th moment). 

\subsection{Problem Statement}
In this paper, we are interested in designing a consensus control law for stochastic multi-agent systems \eqref{MAS_dynamic} such that they achieve consensus as defined in Definition \ref{con_exp} (or \ref{con_as}) and the evolution of the relative positions between neighboring agents should satisfy some prescribed performance bounds in the sense of $q$th moment. Next, we formally define the problem.
\begin{problem}\label{prob}
	Given a multi-agent system defined by \eqref{MAS_dynamic} with the communication graph $\G=(\V,\E)$ and the prescribed performance functions $\rho_k$, $k\in[1;m]$, as in \eqref{performane_function}, derive a distributed control strategy such that the controlled multi-agent system achieves {consensus in expectation (or almost sure consensus)} while satisfying prescribed performance constraints \eqref{performane_bound} in the sense of $q$th moment. 
\end{problem} 
\section{Consensus Control with Prescribed Performance Guarantees}
In this section, we design the control law for the system \eqref{MAS_dynamic} that guarantees consensus while satisfying prescribed performance constraints in the sense of $q$th moment. Here, we assume that the communicating agents can share information about their performance functions $\rho_k$ and transformation functions $T_k$, $k\in[1;m]$. This means the communication between agents is bidirectional and the graph $\G$ is assumed to be undirected. 

We first rewrite the dynamics of the considered multi-agent system in edge space by multiplying \eqref{MAS_dynamic} with $D^T$ on both sides as follows:
\begin{align}\label{MAS_edgespace}
\diff \ol x=(-L_e\ol x+D^T v)\diff t+D^TG(x) \diff W_t,
\end{align}
where $L_e$ is the edge Laplacian which is positive definite if the graph is a tree \cite{dimarogonas2010stability}. 

From \eqref{MAS_edgespace}, the transformed error dynamics in an incremental form is written as follows:
\begin{align}\label{errdyn}
\diff \xi= \big(\Phi_t(-L_e\ol x+D^Tv+\alpha_t\ol x)\big)\diff t+\Phi_t D^TG(x)  \diff W_t,
\end{align}    
where $\xi=[\xi_1,\ldots,\xi_m]^T$, $\Phi_t=diag\{\phi_1(\hat x_1,t),\ldots,$ $\phi_m(\hat x_m,t)\}$, and $\alpha_t=diag\{\alpha_1(t),\ldots\alpha_m(t)\}$. {Now by considering augmented state-space $\eta=[\overline{x}, \xi]^T$, we define an augmented dynamics as:
\begin{align}
	\diff \eta=&\Big(\begin{bmatrix}
	-L_e&0_m\\
	-\phi_t(L_e+\alpha_t)& 0_m
	\end{bmatrix}\eta+\begin{bmatrix}
	I_m\\\phi_t
	\end{bmatrix} D^T v\Big)\diff t\nonumber\\&+\begin{bmatrix}
	I_m\\\phi_t
	\end{bmatrix} D^T G(x)\diff W_t.\label{augmented}
\end{align}  }


In the next theorem, we provide a distributed control law and sufficient conditions over system parameters under which we have a solution to Problem \ref{prob} for consensus in expectation. 
{\begin{theorem}\label{thm}
	Consider the stochastic multi-agent system \eqref{MAS_dynamic} with the communication graph being a tree, the predefined performance functions $\rho_k:\R_0^+\ra\R^+, k\in[1;m]$ with decay rates $\epsilon_k$ as in \eqref{performane_function}, the transformation functions $T_k, k\in[1;m]$ as in \eqref{transformation_fun} with $\ol{\epsilon}:=\max\limits_{k\in[1;m]}\epsilon_k$ satisfying
	\begin{align}\label{A0}
	L_e(I_m+\gamma\Phi_t^{-1}+\Phi_t^2)-\ol{\epsilon}I_m&\geq 0
	\end{align}
	for all $t\in\R_0^+$ and a time-varying distributed control law
	\begin{align}\label{cont1}
	v_i(t)\hspace{-.2em}=\hspace{-.2em}-\hspace{-.2em}\hspace{-.2em}\sum_{k\in\Omega_i}\hspace{-.2em}\hspace{-.2em}\big((\phi_k(\hat x_k,t))^2\ol x_k+\xi_k\big), i\in \V, t\in\R_0^+,
	\end{align} 
	where $\Omega_i=\{k\mid (i,j)=e_k,j\in\mathcal{N}_i\}$, i.e., the set of all the edges that include agent $i\in\V$ as a node. If there exist constants $\kappa,\gamma\in\R^+$ such that   
	\begin{subequations}\label{kappa_condition}
		\begin{align}
		\Phi_t L_e&\geq 2\kappa I_m,\label{A1}\\
		L_e(I_m+\Phi_t^2)-\frac{\mathsf{k}_g^2}{2}(I_m+\frac{1}{\gamma}\Phi_t^2)&\geq 2\kappa I_m ,\label{A2}
		\end{align}
	\end{subequations}
	for all $ t\in\R_0^+$, where $\mathsf{k}_g$ is the Lipschitz constant of the diffusion term $g(\cdot)$ in \eqref{agent_dynamic}, then the controller in \eqref{cont1} achieves consensus in expectation for any initial relative position $\ol x_k(0)\in(-\rho_k(0),\rho_k(0)), \ k\in[1;m]$ while guaranteeing prescribed performance in \eqref{performane_bound} in the sense of $2nd$ moment.
\end{theorem}}
\begin{proof}
	Consider a Lyapunov-like function $V:\R^{2m}\ra\R_0^+$ as
	\begin{align*}
	V(\xi,\ol x)=\frac{1}{2}\xi^T\xi+\frac{\gamma}{2}\ol x^T\ol x,
	\end{align*}
	where $\gamma\in\R_0^+$ is a constant satisfying \eqref{kappa_condition}. One can readily verify that the function $V$ satisfies condition \eqref{abcd} in Lemma \ref{Lemma_Lyapunov} with functions $\ul{\psi}(s):=\frac{1}{2}\min\{1,\gamma\} s$ and $\ol{\psi}(s):=\frac{1}{2}\max\{1,\gamma\} s$ for all $s\in\R_0^+$.
	The corresponding infinitesimal generator as defined in \eqref{infinitesimal} along the multi-agent system in edge space \eqref{MAS_edgespace} and transformed error dynamics \eqref{errdyn} is given by
	\begin{align*}
	\mathcal{L}V(\xi,\ol x)&=\xi^T\Phi_t(-L_e\ol x+D^Tv+\alpha_t\ol x)\\&\hspace{-.2em}+\hspace{-.2em}\gamma\ol x^T(\hspace{-.2em}-\hspace{-.1em}L_e\ol x\hspace{-.2em}+\hspace{-.2em}D^T v)\hspace{-.2em}+\hspace{-.2em}\frac{1}{2}G(x)^TD\Phi_t^T\Phi_t D^T\hspace{-.1em}G(x)\hspace{-.2em}\\&+\hspace{-.2em}\frac{\gamma}{2}G(x)^T\hspace{-.1em}DD^TG(x).
	\end{align*}
	The stack vector of external inputs \eqref{cont1} is written as
	\begin{align}\label{controller1}
	v(t)=-D\Phi_t^2\ol x-D\xi, \ t\in\R_0^+.
	\end{align}
	By substituting control law \eqref{controller1} and using the Lipschitz continuity of $g(x_i)$, we get
	\begin{align}\label{eq_thm1}
	\mathcal{L}&V(\xi,\ol x)\hspace{-.2em}\leq\hspace{-.2em}-\xi^T\Phi_t L_e\ol x\hspace{-.1em}-\hspace{-.1em}\xi^T\Phi_t D^TD\Phi_t^2\ol x-\xi^T\Phi_t D^T D\xi\nonumber\\&+\xi^T\Phi_t\alpha_t\ol x-\gamma\ol x^T L_e\ol x-\gamma\ol x^TD^T D\Phi_t^2\ol x\nonumber\\&-\gamma\ol x^TD^T D\xi+\frac{1}{2}\mathsf{k}_g^2\ol x^T\Phi_t^2 \ol x+\frac{\gamma}{2}\mathsf{k}_g^2(\ol x^T\ol x)
	\end{align} 
	Since the considered graph is a tree, we know that the edge Laplacian $L_e=D^TD$ is positive definite. 
	We also know from the fact that $\alpha_k(t)=-\frac{\dot{\rho_k}(t)}{\rho_k(t)}>0$ and \eqref{performane_function} that $\alpha_k(t)<\epsilon_k$ for all $t\in\R_0^+$ and hence $\alpha_t<\ol\epsilon:=\max\limits_{k\in[1;m]}\epsilon_k$. With the aforementioned facts, inequality \eqref{eq_thm1} reduces to 
	\begin{align*}
	\mathcal{L}&V(\xi,\ol x)\leq\ \xi^T\big(\Phi_t(-L_e(I_m+\gamma\Phi_t^{-1}+\Phi_t^2)+\ol\epsilon I_m)\big)\ol x\nonumber
	\\&-\gamma\ol x^T\big(L_e(I_m+\Phi_t^2)-\frac{\mathsf{k}_g^2}{2}(I_m+\frac{1}{\gamma}\Phi_t^2)\big)\ol x-\xi^T\Phi_t L_e\xi.
	\end{align*}
	Note that $\Phi_t$ is a positive definite matrix and $\xi_k(\ol x_k/\rho_k)$ is strictly increasing with $\xi_k(0)=0$ which implies that $\xi_k(\ol x_k/\rho_k)\ol x_k\geq 0$. Thus, by using condition \eqref{A0}, we can readily verify that the first term is non-positive.  
	Further, by following conditions in \eqref{kappa_condition}, one obtains $$\mathcal{L}V(\xi,\ol x)\leq -\kappa V(\xi,\ol x)$$ 
	with a constant $\kappa\in\R^+$ satisfying \eqref{kappa_condition}. Now by following the result of Lemma \ref{Lemma_Lyapunov}, one ensures the consensus in expectation; and since  $\ol x_k(0)\in(-\rho_k(0),\rho_k(0)), \ k\in[1;m]$ and $\EE[\|\xi(t)\|^2]$ is bounded for all $t\in\R_0^+$, we ensure that the $\ol x$ satisfies prescribed transient constraints \eqref{performane_bound} in the sense of $2nd$ moment.
	This concludes the proof.  
\end{proof}
\begin{remark}
		Note that the proposed time-varying distributed control law \eqref{cont1} applied to agents is the composition of the term based on the prescribed performance and the relative positions of the neighbours.
\end{remark}
\begin{remark}
	Since the term $\phi_k(\hat{x}_k,t)$ is lower bounded by $\min\limits_{t\in\R_0^+,\hat x_k\in(-1,1)}\frac{1}{\rho_k(t)}\frac{2}{1-\hat{x}_k^2}=\frac{2}{\rho_{k0}}$, $k\in[1;m]$ for all $t\in\R_0^+$, one can find a upper bound for $\kappa$ satisfying inequalities \eqref{kappa_condition} for appropriate values of $\gamma$. From \eqref{A2}, we would like to emphasize that smaller values of $\kappa$ imply that one can handle stronger noise, i.e., bigger value of $\mathsf{k}_g$. Moreover, since one can always find a constant $\gamma$ for any value of $\ol{\epsilon}$ satisfying \eqref{A0}, we can use the proposed controller to have a result for prescribed performance functions with any decay rate. 
\end{remark}
{The next theorem provides result for ensuring almost sure consensus while considering performance constraints. 
\begin{theorem}\label{thm2}
	Consider the stochastic multi-agent system \eqref{MAS_dynamic} with the communication graph being a tree, the predefined performance functions $\rho_k:\R_0^+\ra\R^+, k\in[1;m]$ as in \eqref{performane_function}, and a time-varying distributed control law 
	\begin{align}\label{control2}
		v_i(t)\hspace{-.2em}=\hspace{-.2em}-\hspace{-.2em}\hspace{-.2em}\sum_{k\in\Omega_i}(\phi_k(\hat x_k,t))^2\ol x_k, i\in \V, t\in\R_0^+.
	\end{align}
	If 
	\begin{align}
		L_e-\frac{1}{2}k_g^2I_m>0,
	\end{align}
	where $\mathsf{k}_g$ is the Lipschitz constant of the diffusion term $g(\cdot)$ in \eqref{agent_dynamic}, then the controller in \eqref{control2} achieves almost sure consensus for any initial relative position $\ol x_k(0)\in(-\rho_k(0),\rho_k(0)), \ k\in[1;m]$ while guaranteeing prescribed performance in \eqref{performane_bound} in the sense of $q$th moment, where $q\in\{1,2\}$.
\end{theorem}
\begin{proof}
	Consider a Lyapunov-like function $V:\R^{2m}\ra\R_0^+$ as
	\begin{align*}
	V(\eta)=\Big(\frac{1}{q} \eta^T\eta\Big)^\frac{q}{2},
	\end{align*}
	where $q\in\{1,2\}$. One can readily verify that the function $V$ satisfies condition \eqref{abcd} with functions $\ul{\psi}(s):=(\frac{1}{q})^{\frac{q}{2}} s$ and $\ol{\psi}(s):=(\frac{2m}{q})^{\frac{q}{2}} s$ for all $s\in\R_0^+$.
	The corresponding infinitesimal generator as defined in \eqref{infinitesimal} along the augmented dynamics \eqref{augmented} of the multi-agent system in edge space \eqref{MAS_edgespace} and transformed error dynamics \eqref{errdyn} is given by
	\begin{align*}
	&\mathcal{L}V(\eta)=\eta^T\Big(\frac{1}{q} \eta^T\eta\Big)^{\frac{q}{2}-1}\Big(\begin{bmatrix}
	-L_e&0_m\\
	-\phi_t(L_e+\alpha_t)& 0_m
	\end{bmatrix}\eta\\&+\begin{bmatrix}
	I_m\\\phi_t
	\end{bmatrix} D^T v\Big) + \frac{1}{2} G(x)^TD\begin{bmatrix}
	I_m\\\Phi_t
	\end{bmatrix}^T\Bigg(\Big(\frac{1}{q} \eta^T\eta\Big)^{\frac{q}{2}-1}I_{2m}\\&+\frac{q-2}{q}\eta\eta^T\Big(\frac{1}{q} \eta^T\eta\Big)^{\frac{q}{2}-2}\Bigg) \begin{bmatrix}
	I_m\\\phi_t
	\end{bmatrix} D^T G(x).
	\end{align*}
	The stack vector of external inputs is written as
	\begin{align}\label{controller2}
	v(t)=-D\Phi_t^2\ol x, \ t\in\R_0^+.
	\end{align}
	By substituting control law \eqref{controller2} and using the Lipschitz continuity of $g(x_i)$ and the fact that $q\in\{1,2\}$, the $\mathcal{L}V(\eta)$ is rewritten as
	\begin{align*}
		&\mathcal{L}V(\eta)\leq\eta^T\Big(\frac{1}{q} \eta^T\eta\Big)^{\frac{q}{2}-1}\Big(\begin{bmatrix}
	-L_e&0_m\\
	-\Phi_t(L_e+\alpha_t)& 0_m
	\end{bmatrix}\eta\\&\ -\begin{bmatrix}
	I_m\\\Phi_t
	\end{bmatrix} D^T D\Phi_t^2\ol x\Big) + \frac{1}{2} k_g^2\ol x^T\Big(\frac{1}{q} \eta^T\eta\Big)^{\frac{q}{2}-1}\begin{bmatrix}
	I_m\\\Phi_t
	\end{bmatrix}^T\begin{bmatrix}
	I_m\\\Phi_t
	\end{bmatrix} \ol x.\\
	&=\eta^T\Big(\frac{1}{q} \eta^T\eta\Big)^{\frac{q}{2}-1}\hspace{-.2em}\Big(\hspace{-.2em}\begin{bmatrix}
	-L_e&0_m\\
	-\Phi_t(L_e\hspace{-.2em}+\hspace{-.2em}\alpha_t)& 0_m
	\end{bmatrix}\hspace{-.2em}\eta\hspace{-.2em}-\hspace{-.2em}
	\begin{bmatrix}
	I_m\\\Phi_t
	\end{bmatrix} \hspace{-.2em}
	L_e
	\begin{bmatrix}\hspace{-.2em}
	\Phi_t^2\\0_m
	\end{bmatrix}\hspace{-.1em}\eta\hspace{-.2em}\Big) \\&\ + \frac{1}{2} k_g^2\Big(\frac{1}{q} \eta^T\eta\Big)^{\frac{q}{2}-1}\eta^T\begin{bmatrix}
	\begin{bmatrix}
	I_m\\\Phi_t
	\end{bmatrix}^T\begin{bmatrix}
	I_m\\\phi_t
	\end{bmatrix} &0_m\\
	0_m & 0_m
	\end{bmatrix} \eta\\&
	=\eta^T\Big(\frac{1}{q} \eta^T\eta\Big)^{\frac{q}{2}-1}\begin{bmatrix}
	-L_e(I_m\hspace{-.2em}+\hspace{-.2em}\Phi_t^2)+\frac{1}{2}k_g^2(I_m\hspace{-.2em}+\hspace{-.2em}\Phi_t^2)&0_m\\
	-\Phi_t(L_e+\alpha_t+L_e\Phi_t^2)& 0_m
	\end{bmatrix}\hspace{-.2em}\eta\\&
	=-\eta^T\Big(\frac{1}{q} \eta^T\eta\Big)^{\frac{q}{2}-1}\begin{bmatrix}
	(L_e-\frac{1}{2}k_g^2I_m)(I_m+\Phi_t^2)&0_m\\
	\Phi_t(L_e+\alpha_t+L_e\Phi_t^2)& 0_m
	\end{bmatrix}\eta\\&=-\beta(\eta).
	\end{align*}
	Note that $\Phi_t$ and $L_e$ is a positive definite matrices. Hence, if $L_e-\frac{1}{2}k_g^2I_m>0$, we get $\mathcal{L}V(\eta)\leq 0$. This implies that $\mathbb{E}[\|\eta(t)\|^q]$ is bounded and hence  $\mathbb{E}[\|\xi(t)\|^q]$ is bounded for all $t\in\R_0^+$ and all $\ol x_k(0)\in(-\rho_k(0),\rho_k(0))$, $k\in[1;m]$. This ensures that the $\ol x$ satisfies prescribed transient constraints \eqref{performane_bound} in the sense of $q$th moment, where $q\in\{1,2\}$. Now, since $\beta(x)$ is continuous and nonnegative, by utilizing stochastic Barbalat's lemma (Lemma \ref{stoch_barbalat}), we have $\lim\limits_{t\ra\infty}\beta(\eta(t))=0$ almost surely. This implies almost sure consensus (i.e. $\lim\limits_{t\ra\infty}\ol x(t)=0$ a.s.).  
\end{proof}}

 \begin{figure}[t]
	\centering
	\includegraphics[scale=0.045]{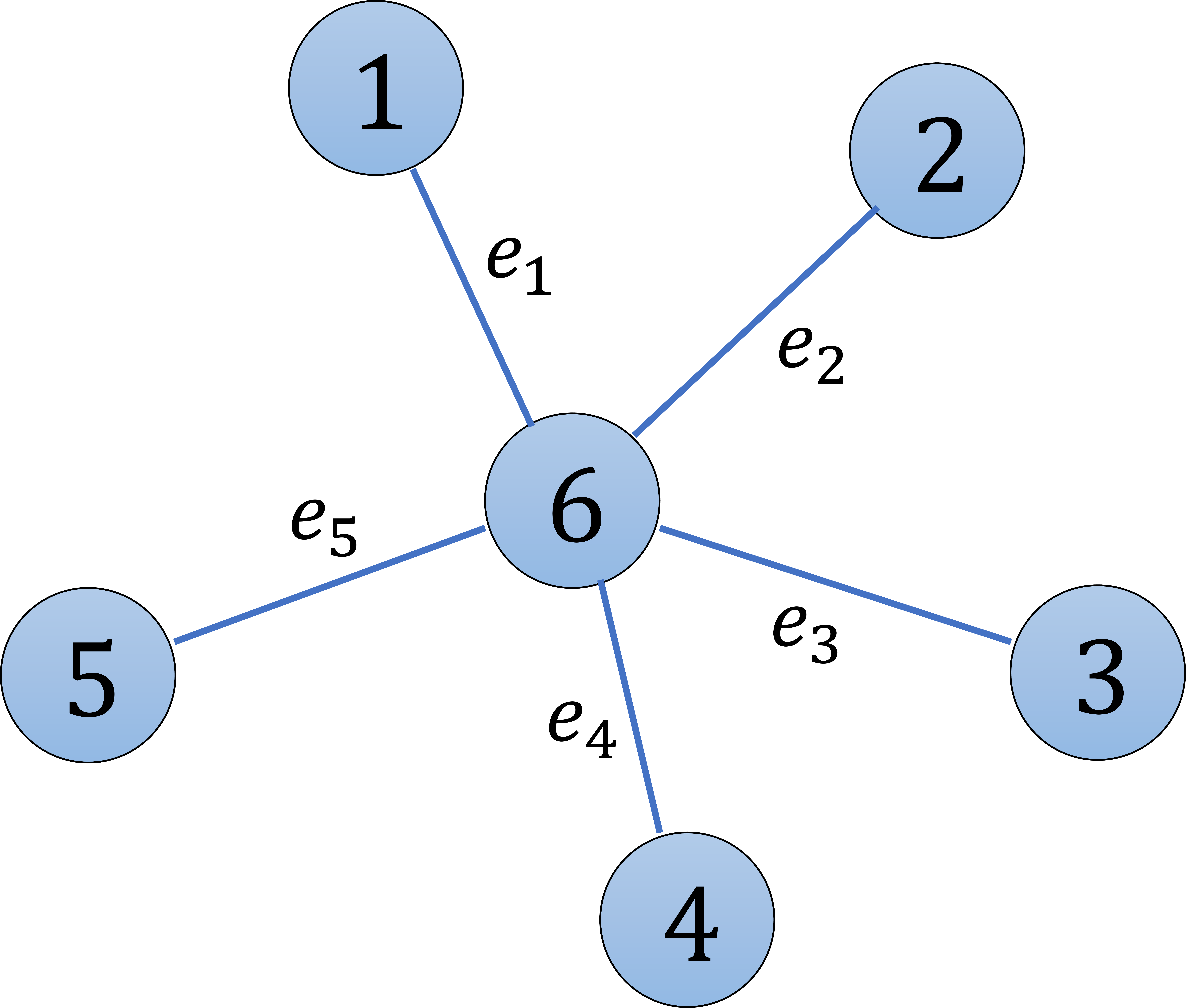}
	\caption{Communication graph with tree topology}
	\label{fig:star}
\end{figure} 
\begin{figure*}[t]
	\centering
	\includegraphics[scale=0.37]{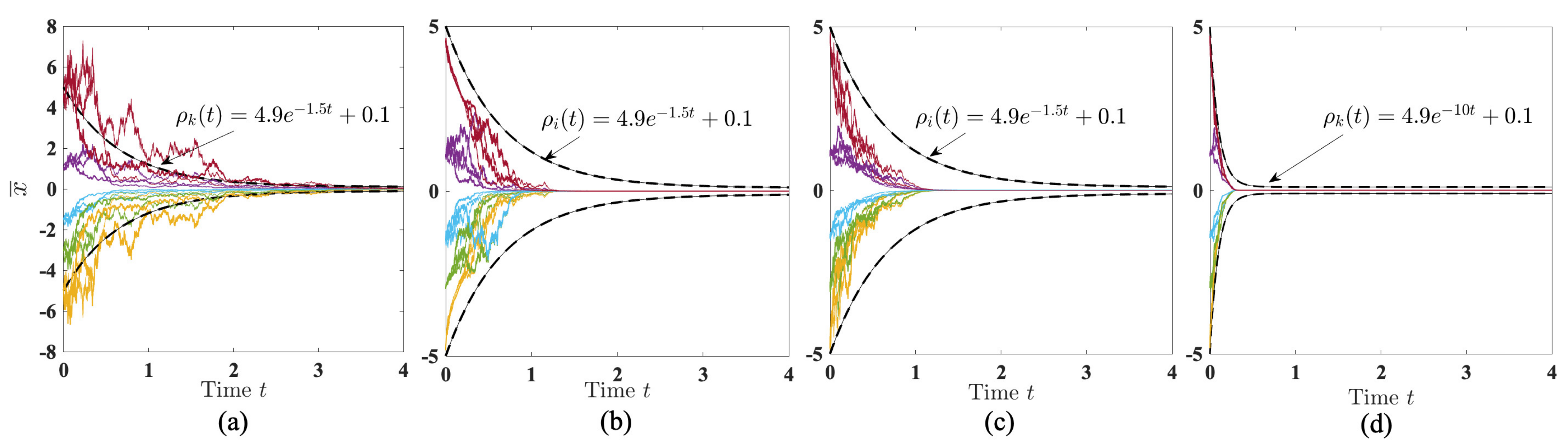}
	\caption{Several realizations of relative positions of stochastic multi-agent system \eqref{num_ex} starting from $\ol x(0)=[-4.9,1,-3,-1.5,4.5]^T$ $(a)$ without prescribed performance control law; $(b)$ with the proposed control law \eqref{control2}; $(c)$ and $(d)$ are with the proposed control law \eqref{cont1} and decay rates of prescribed performance functions $\epsilon_k=1.5$ and $\epsilon_k=10$ $\forall k\in[1;5]$, respectively. The dashed black lines represents prescribed performance constraints.}
	\label{fig:PPC}
\end{figure*} 
\section{Numerical Examples}
 \begin{figure}[t]
	\centering
	\includegraphics[scale=0.45]{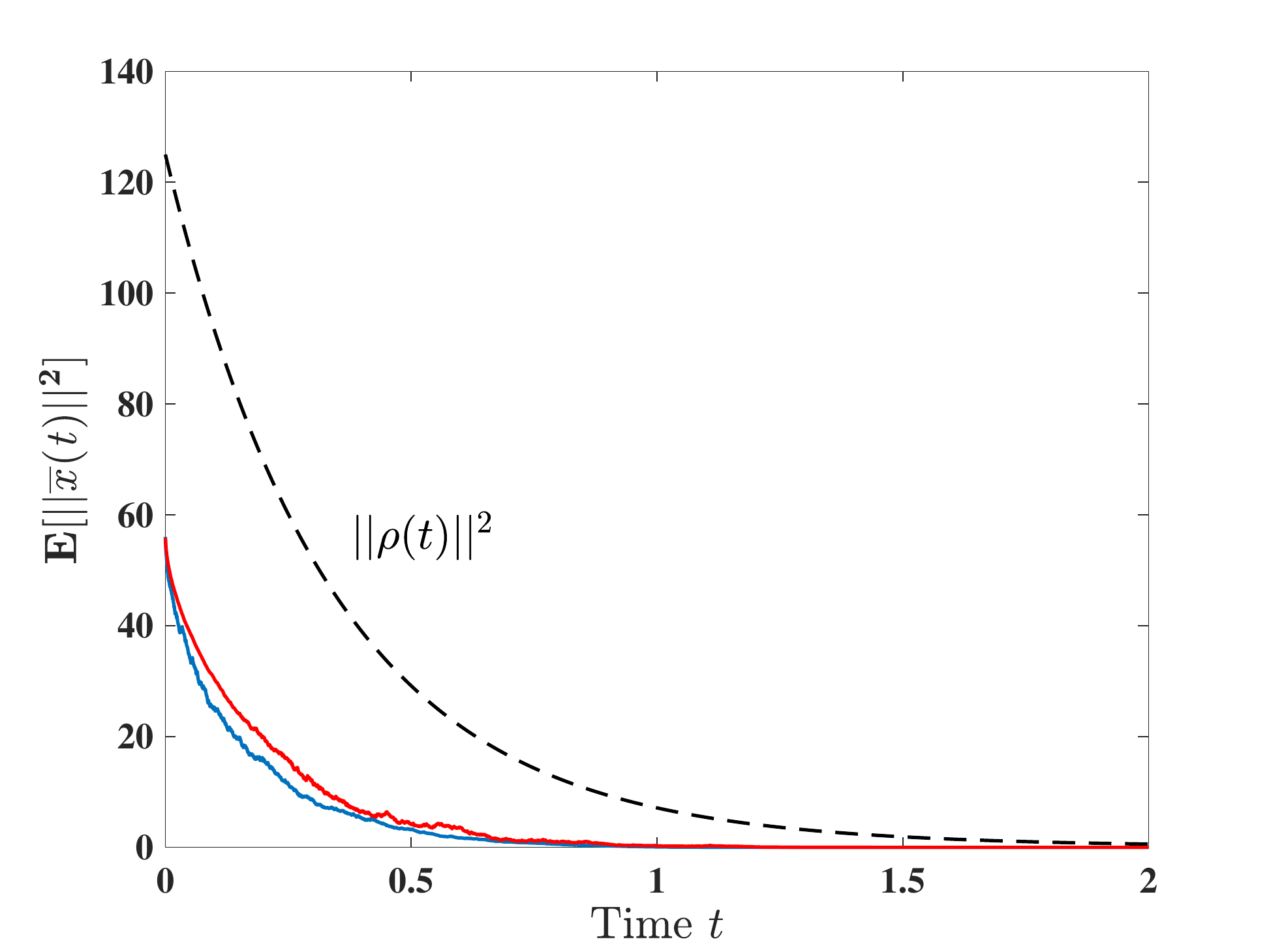}
	\caption{Mean-squared values of relative position starting from $\ol x(0)=[-4.9,1,-3,-1.5,4.5]^T$ under proposed control laws \eqref{cont1} [blue] and \eqref{control2} [red]. The dashed black line shows squared values of performance function $\rho(t)=[\rho_1(t),\ldots,\rho_m(t)]^T$, $\rho_k(t)=4.9\e^{-1.5t}+0.1$ for all $k\in[1;5]$.}
	\label{fig:star1}
\end{figure} 
In order to demonstrate the effectiveness of the proposed results, we consider a simulation example of a stochastic multi-agent system consist of six agents given by stochastic differential equations as 
\begin{align}\label{num_ex}
	\diff x_i=u_i \diff t+\e^{-0.1|x_i|}\sin(x_i) \diff W_t, \ i\in[1;6],
\end{align} 
where $x_i\in\R$ and $u_i\in\R$ are the absolute position and the control input of the $i$th agent, respectively, $W_t$ is the standard Brownian motion, and the diffusion term $g(x_i)=\e^{-0.1|x_i|}\sin(x_i)$ with the corresponding Lipschitz constant $\mathsf{k}_g=1$. The undirected communication graph with $\V=\{1,2,3,4,5,6\}$ and $\E=\{e_1,e_2,e_3,e_4,e_5\}$ (\ie, $n=6$ and $m=5$) is shown in Figure \ref{fig:star}. We consider the prescribed performance functions $\rho_k(t)=4.5\e^{-\epsilon t}+0.1$ for all $k\in[1;m]$ and we consider two cases: $\epsilon=1.5$ and $\epsilon=10$. The performance bounds $-\rho_k(t)$ and $\rho_k(t)$ are depicted with black dashed lines in Figure \ref{fig:PPC}. We obtain $\gamma=4$ and $\kappa=0.39$ satisfying conditions \eqref{A1}, and \eqref{A2} in both the cases. Figure \ref{fig:PPC} shows several realizations of relative positions of the controlled stochastic multi-agent systems starting from $\ol x(0)=[-4.9,1,-3,-1.5,4.5]^T$  without an external input (Figure \ref{fig:PPC}(a)) and with the proposed distributed controller \eqref{control2} for decay rate $\epsilon=1.5$ (Figure \ref{fig:PPC}(b)), and with the proposed distributed controller \eqref{cont1} with different values of decay rates of performance bounds (Figure \ref{fig:PPC}(c) and (d)). The mean-squared value of relative positions of the controlled stochastic multi-agent system $\EE[\|\ol x(t)\|^2]$ computed for 1000 realizations is shown in Figure \ref{fig:star1}. From Figures \ref{fig:PPC} and \ref{fig:star1}, one can readily verify that the proposed control laws achieve almost sure consensus and consensus in expectation, respectively, while respecting prescribed performance constraints in $2nd$ moment.     

\section{Conclusion}\label{conclusion}
In this work, we studied a consensus problem of stochastic multi-agent systems with prescribed performance bounds. Under the assumption of a tree graph, a distributed control law has been proposed for multi-agent systems containing agents with state-dependent stochastic noise such that the entire system can achieve consensus in expectation (or almost surely) while satisfying predefined performance constraints in the sense of $q$th moment.     
Future work includes extending the results for more general graphs with cycles and for heterogeneous agents.


\bibliographystyle{IEEEtran}
\bibliography{bibliography}

\addtolength{\textheight}{-12cm}   

\end{document}